\documentclass[preprint,authoryear,10pt]{elsarticle}
\usepackage{url}
\usepackage{lifecon}
\usepackage{prettyref}
\usepackage{ctable}
\usepackage{threeparttable}
\usepackage{booktabs}
\usepackage{subfig}
\usepackage{todonotes}\usepackage{color}\newcommand{\revised}[1]{\textcolor{black}{#1}}
\DeclareMathOperator{\al}{\alpha}

\DeclareMathOperator{\faira}{\alpha^\star}
\DeclareMathOperator{\Q}{\mathbb{Q}}

\DeclareMathOperator{\ban}{\bar{a}}

\DeclareMathOperator{\bartau0}{\bar{\tau}_{0}}
\newcommand{\ann}[2]{#1_{\lcroof{#2}}}
\newcommand{\mb}[1]{\mathbb{#1}}
\newcommand{\mc}[1]{\mathcal{#1}}

\newrefformat{tab}{Table\,\,\ref{#1}}
\newrefformat{fig}{Figure\,\,\ref{#1}}
\newrefformat{eq}{(\ref{#1})}
\newrefformat{thm}{Theorem\,\,\ref{#1}}
\newrefformat{def}{Definition\,\,\ref{#1}}
\newrefformat{rem}{Remark\,\,\ref{#1}}
\newrefformat{lem}{Lemma\,\,\ref{#1}}
\newrefformat{pro}{Proposition\,\,\ref{#1}}
\newrefformat{ass}{Assumption\,\,\ref{#1}}
\newrefformat{exa}{Example\,\,\ref{#1}}
\newrefformat{sec}{Section\,\,\ref{#1}}
\newrefformat{sub}{Subsection\,\,\ref{#1}}
\newrefformat{cha}{Chapter\,\,\ref{#1}}
\usepackage{amssymb}
 \usepackage{amsthm}

\theoremstyle{plain}
\ifx\thechapter\undefined
\newtheorem{thm}{Theorem}
\else
\newtheorem{thm}{Theorem}[chapter]
\fi
  \theoremstyle{plain}
  \newtheorem{assumption}[thm]{Assumption}
  \theoremstyle{plain}
  \newtheorem{lem}[thm]{Lemma}
  \theoremstyle{definition}
  \newtheorem{defn}[thm]{Definition}
  \theoremstyle{remark}
  
  \theoremstyle{remark}
  \newtheorem*{remUnnum}{Remark}
  \theoremstyle{plain}
  \newtheorem{prop}[thm]{Proposition}
  \theoremstyle{remark}
  
 \theoremstyle{definition}
  
  \theoremstyle{plain}
  
  \theoremstyle{plain}
  \newtheorem{cor}[thm]{Corollary}

\journal{}

\begin{document}

\begin{frontmatter}

 \title{Valuation Perspectives and Decompositions for Variable Annuities with GMWB riders}
 \author[Conc]{Cody~B.~Hyndman\corref{cor1}}
\ead{cody.hyndman@concordia.ca}

\author[Conc,Guard]{Menachem~Wenger\corref{cor2}}
\ead{menachem\_wenger@glic.com}
\cortext[cor1]{Corresponding author}
\cortext[cor2]{The views and opinions expressed in this paper are those of the individual author(s) and do not necessarily reflect the views of The Guardian Life Insurance Company of America.}
\address[Conc]{Department of Mathematics and Statistics, Concordia University, 1455 De Maisonneuve Blvd.\ W., Montr\'eal, Qu\'ebec, Canada H3G 1M8}
\address[Guard]{The Guardian Life Insurance Company of America, New York, NY}

\begin{abstract}
The guaranteed minimum withdrawal benefit (GMWB) rider, as an add on to a variable annuity (VA), guarantees the return of premiums in the form of periodic withdrawals while allowing policyholders to participate fully in any market gains. GMWB riders represent an embedded option on the account value with
a fee structure that is different from typical financial derivatives.  We consider fair
pricing of the GMWB rider from a financial economic perspective.  Particular focus is placed on the 
\revised{distinct} perspectives of the insurer and policyholder and the unifying relationship. 
We extend a decomposition of the VA contract into components that reflect term-certain payments 
and embedded derivatives to the case where the policyholder has the option to surrender, or lapse, the contract early.

\end{abstract}


\end{frontmatter}

\section{\label{sec:Introduction}Introduction}

\subsection{\label{sub:Background}Background}

Variable deferred annuities have two phases: the accumulation period
and the annuitization period. During the accumulation period, premiums
are deposited with the insurer and can be actively managed by the
policyholder to achieve investment goals by allocating the funds
to a selection of investment funds. The policyholder may choose to
take partial withdrawals and/or surrender the contract, although the
proceeds will likely be subject to contingent deferred sales charges
(CDSC) and possible
tax penalties depending on the age of the policyholder. Upon annuitization
the policyholder cedes control over the funds and in return is guaranteed
a periodic stream of payments, ranging from a fixed number of years up to guaranteed for life. 
This phase protects annuitants from longevity risk.

Riders are optional add-ons to variable annuities (VA), 
providing additional benefits in return for which an additional
charge is subtracted annually from the account value (AV). 
The first riders introduced to the VA market were death benefit riders:
these guarantee a minimum death benefit to the beneficiaries if the
policyholder dies during the accumulation period. Death benefits evolved 
from a simple return of premium to increasingly
rich guarantees in the form of annual roll-ups and highest anniversary
values. The next form of riders introduced were guaranteed living
benefits (GLBs). The guaranteed minimum accumulation benefit riders
(GMABs) guarantee a minimum account value at a specific date (i.e.,
10 years from issue date), while the guaranteed minimum income benefit
riders (GMIBs) guarantee a minimum annuitization amount by giving
policyholders the choice between annuitizing a higher guarantee base
at contractually specified annuitization rates or the current account
value at the current annuitization rates.

Guaranteed minimum withdrawal benefit riders (GMWBs) were introduced
in 2002 and guarantee the policyholder will recover at least the total
premiums paid into the policy in the form of periodic withdrawals,
subject to the annual withdrawals not exceeding a contractual percentage
of the premiums. By allowing policyholders to remain in the accumulation
phase and retain full control of their investments, policyholders
benefit from the upside potential of equity investments while being protected
from downside risk. GMWBs evolved into the guaranteed lifetime withdrawal
benefit riders (GLWBs) which guarantee the annual maximal withdrawals
for life. GMWB and GLWB riders represent embedded
financial put options on the account values and techniques from mathematical
finance have been used to value these contracts.

The fee structures of VA riders add complexity to pricing and risk
management processes relative to the standard financial equity market
derivatives.  For a standard financial option a single upfront premium 
is charged which has no impact on the future random payoffs. 
No upfront fees are charged for GMWB riders but instead fees are deducted 
periodically from the AV to pay for the rider where the
fees are proportional to the AV. The AV is influenced by the withdrawal
behaviour of the policyholder and fee revenue stops in
the event of death or surrender. As such there are multiple sources
of uncertainty surrounding the fees. As well, an increase in the fee
rate results in higher fee income but it also creates a drag
on the AV, potentially causing it to reach zero faster which would result
in earlier termination of fee revenues and increased rider guarantee
payouts.

It is our belief that the GMWB and GLWB riders play a more prominent role 
in driving sales of VAs and their accompanying profits rather than as a source of direct profit for insurers.
Shortly, we point to the consensus among the early papers that these
riders were underpriced, supporting this hypothesis that they were
only a means to increase VA sales. Indeed,
reinsuring all or most of the risk was a popular risk management strategy for the initial
GMWB products. Reinsurance premiums increased as reinsurers became
more informed of the high risk embedded in these products. Around
the time of the financial crisis in 2008 reinsurers stopped offering
coverage altogether on GMWB and GLWB riders at which point the importance
of internal dynamic hedging programs rose rapidly.

With this in mind, we consider pricing and hedging the GMWB product
in a simplified framework consistent with the no-arbitrage principle
from financial economics. GLB riders have grown increasingly complex in recent years. 
Added features range from periodic ratchets and annual roll-ups to specific
one-time bonuses if certain criteria are met. These features
were designed to both increase the product appeal and entice policyholders
to delay withdrawals to the benefit of the insurer. It is evident that the GLWB riders have
come to define the VA market. GMWB riders were the precursor to the
GLWBs and as such, a mathematical analysis of the GMWB product is
interesting in its own right, and can provide insights that may be
extended to more complex products.

\subsection{\label{sub:ProductSpec}Product Specifications and Notation}

An underlying VA contract plus a GMWB rider is issued at time $t=0$ to
a policyholder of age $x$ and an initial premium $P$ is received.
We assume no subsequent premiums. The premium is invested into a fund
which perfectly tracks a risky asset $S=\{S_{t};t\geq0\}$ with no
basis risk. The rider fee rate $\al$ is applied to the account value $W=\{W_{t};t\geq0\}$.
Fees are deducted from the account value as long as the contract is in force and the
account value is positive. 

A guaranteed maximal withdrawal rate $g$ is contractually specified
and the rider specifies that up to the amount $G:=gP$ can be withdrawn annually until $P$ is recovered, regardless of the evolution of $\{W_{t}\}$.
If the account value
hits zero, then the policyholder receives withdrawals at rate $G$
until the initial premium has been recovered. Policyholders always have the option of withdrawing any amount provided it
does not exceed the remaining account value.  If annual withdrawals
exceed $G$ while the account value is still positive, then a surrender
charge is applied to the withdrawals and a reset feature may reduce
the guarantee value, e.g., the remaining portion of the initial premium
not yet recovered. Policyholders also have the option of surrendering early (the terminology of lapses and surrenders are used interchangeably) and receiving the account value less a surrender charge. Any guarantee value is forfeited by surrendering. 

Assuming a static withdrawal strategy where $G$ is withdrawn annually 
 we set the maturity $T:=1/g$ so that cumulative withdrawals at time $T$ equal $P$. 
At time $T$ the
rider guarantee is worthless and the policyholder receives a terminal
payoff of the remaining account value, if it is positive. 
This assumption translates over to a real-world trend of no annuitizations and is justified since a high proportion of VAs are not annuitized.

\subsection{\label{sub:LiterRev}Literature Review}

The initial paper on pricing and hedging GMWB products,
by \citep{RefWorks:7}, employs continuous withdrawals and a standard geometric Brownian motion model for $\{S_{t}\}$ and considers two policyholder behaviour strategies. Under a static withdrawal strategy and no lapses
the contract is decomposed into a term certain component and a Quanto
Asian Put option. Numerical PDE methods are used to evaluate the ruin probabilities
for $\{W_{t}\}$ and the contract value $V_{0}$. A dynamic behaviour
strategy is considered where optimal withdrawals occur. This free boundary value
problem is solved numerically for $V_{0}$. It is found that the
optimal strategy reduces to withdrawing $G$ continuously unless $W_{t}$
exceeds a boundary value depending on the remaining guarantee balance
of $P-Gt$, in which case an arbitrarily large withdrawal rate is
taken and the policyholder should lapse. \citep{RefWorks:7} conclude
that the GMWB riders in effect in 2004 were underpriced.

The optimal behaviour approach is formalized in \citep{RefWorks:8}
where a singular stochastic control problem is posed. Unlike in \citep{RefWorks:7},
time dependency and a complete description of the auxiliary conditions
are included in this model. To facilitate numerical solutions for
the HJB equations a penalty approximation formulation is solved using
finite difference methods.  Consistent with \citep{RefWorks:7}, numerical results provide support  that the provision for optimal behaviour is quite valuable and insurers appeared to be underpricing GMWB riders. The optimal strategy consists of withdrawing at rate $G$ (continuously) except for in certain regions
of the state space where an infinite withdrawal rate is optimal, which
means to {}``withdraw an appropriate finite amount instantaneously
making the equity value of the personal account and guarantee balance
to fall to the level that it becomes optimal for him to withdraw {[}$G${]}''
(\citealp{RefWorks:8}). However, \citep{RefWorks:8} allow the policyholder
the option of withdrawing any amount of the unrecovered initial premium,
even if it exceeds the account value. In other words, if the account
value is zero, the policyholder can elect to receive the remaining
guarantee balance instantly subject to surrender charges rather than
receive $G$ annually. The impact of this assumption is amplified
by not including a reset feature in most of their work. The combination
of this is the main cause of arriving at optimal strategies differing
from \citep{RefWorks:7}. 

\citep{RefWorks:14} extend \citep{RefWorks:8} to an impulse control
problem representation where the control set allows for continuous
withdrawal rates not exceeding $G$ and instantaneous finite withdrawals.
This extension allows for modelling of more complex product features.

\citep{RefWorks:68} develop an extensive and comprehensive framework
to price any of the common guarantees available with VAs, assuming
that any policyholder events such as surrenders, withdrawals, or death
occurs at the end of the year. Deterministic mortality is assumed.
Monte-Carlo simulation is used to price the contracts assuming a deterministic
behaviour strategy for the policyholders. To price the contracts assuming
an optimal withdrawal strategy, a quasi-analytic integral solution
is derived and an algorithm is developed by approximating the integrals
using a multidimensional discretization approach via a finite mesh.
Hence, only a finite subset of all possible strategies are considered.
One drawback is that the valuation with optimal behaviour for a single
contract could take excessive computation time. 

Allowing for discrete withdrawals, \citep{RefWorks:61} consider a
number of guarantees under a more general financial model with stochastic
interest rates and stochastic volatility in addition to stochastic
mortality. In particular for GMWBs, a static behaviour strategy ($G$
withdrawn annually and no lapses) is priced using standard Monte Carlo
whereas an optimal lapse approach ($G$ withdrawn annually) is priced
with a Least Squares Monte Carlo algorithm. 

Upper and lower bounds on the price process for the GMWB are derived
in \citep{RefWorks:32} under stochastic interest rates and assuming
a static continuous withdrawal strategy of $G$ per year with no lapses.
\citep{RefWorks:32} also present a tangential result about the 
relationship between the insured and insurer perspectives which we generalize  
to include early surrender.

Ignoring mortality and working with a static withdrawal assumption
and no lapses, the primary focus of \citep{RefWorks:43} is on developing
semi-static hedging strategies under both a geometric Brownian motion
model and a Heston stochastic volatility model for the underlying
asset $\{S_{t}\}$. However, sufficient attention and detail is paid
to pricing the GMWB rider assuming the insured takes constant withdrawals
of $G/n$ at the end of each period where there are $n$ time steps
per year. \citep{RefWorks:43} observes that the contract (GMWB plus
VA) can be decomposed into a term certain component and a floating
strike Asian Call option on a modified process. Both a Monte Carlo
approach and a moment-matching log-normal approximation method (based
on \citealp{RefWorks:69}) are used to obtain results for increasing
$n$. 

\subsection{Overview}

\revised{The literature covers a wide range of theoretical and numerical approaches to modelling variable annuities with GMWB riders.  The models surveyed differ most importantly in how they integrate model complexity and policyholder behaviour.  Our approach begins with the model presented by \citep{RefWorks:7} and establishes several new analytical results and relationships.}

Based on the product specifications listed in \prettyref{sub:ProductSpec}, which shall be assumed throughout this paper,  optimal withdrawal behaviour reduces to withdrawing at rate $G$ or lapsing. The rider guarantee represents an intangible amount. Once the account value is zero, this amount is accessible
only through withdrawals at rate $G$, a product specification adopted
by both \citep{RefWorks:7} and \citep{RefWorks:61}. The work of
\citep{RefWorks:8} and \citep{RefWorks:14} do not reflect this and
therefore different results are obtained.

In \prettyref{sec:Valuation-of-GMWBS cts} we motivate our work by reviewing the continuous-time model of \citep{RefWorks:7}.  
In \prettyref{sec:val-persp} we formalize the relationship between the value processes for the GMWB rider from the point of view of both the insured and the insurer. 
We start with a restricted model which accounts for equity risk only.  We provide our first results on the existence and uniqueness of a fair fee and a decomposition of the contract into the account value and the guarantee in the no lapse case.  In \prettyref{sec:Extending-modelSurrenders} we extend the model to incorporate early surrenders, that is we incorporate policyholder behaviour risk, and subsequently generalize the decomposition of the value of the contract into the account value, the value of the guarantee under no-lapses, and the value of the option to lapse.  \prettyref{sec:conclusions} concludes and an appendix contains additional technical results and proofs.

\section{Valuation of GMWBs in a Continuous-Time Framework}
\label{sec:Valuation-of-GMWBS cts}
We first consider the continuous model introduced by \citep{RefWorks:7}. 
This will motivate the developments
in the following sections.

Let $(\Omega,\mc{F},\mb{P})$ be a complete probability space
where $\{B_{t}^{'}\}_{0\leq t\leq T}$ is a one dimensional standard
Brownian motion and $T<\infty$.
Define $\mc{F}_{t}:=\sigma\{B_{s}^{'};0\leq s\leq t\},$ for
all $t\in[0,T]$. Consider the financial market consisting of one
risky asset and one riskless asset. 
The unit price of the risky asset $\{S_{t}^{x,u}\}_{u\leq t\leq T}$
follows the geometric Brownian motion process\begin{equation}
dS_{t}=\mu S_{t}dt+\sigma S_{t}dB_{t}^{'},\quad t\geq u,\quad S_{u}=x.\label{eq:StP}\end{equation}
We assume a constant riskless rate $r$; therefore, the riskless asset is the money market account $M_{t}=e^{rt}$ 
for all $t\geq0$.

Applying Girsanov's theorem for Brownian motion, $\{S_{t}^{x,u}\}_{u\leq t\leq T}$
follows the process: \begin{equation}
dS_{t}=rS_{t}dt+\sigma S_{t}dB_{t},\quad t\geq u,\quad S_{u}=x.\label{eq:StQ}\end{equation}
where $\{B_{t}\}$ is a standard Brownian motion under the unique risk neutral measure
$\Q$ equivalent to $\mb{P}$. 
We work with the filtered probability space $(\Omega,\mc{F}_{T},\mb{F},\Q)$
where $\mb{F}=\{\mc{F}_{s}\}_{0\leq s\leq T}$. The financial market is complete and arbitrage-free.

We formulate our initial assumptions as follows. 
\begin{assumption}
\label{ass:Cts_nolapses} We assume a static withdrawal strategy where the policyholder takes
continuous withdrawals at a rate of $G:=gP$ per year. The maturity
is $T:=\frac{1}{g}$ years. Early lapses are not permitted. We also
assume $r>0$. 
\end{assumption}
The account value process $\{W_{t}\}$ is reduced by the instantaneous
rider fees $\al W_{t}dt$ and the instantaneous withdrawals $Gdt$.
By \prettyref{eq:StQ} the account value $W_{t}^{P,0}$ follows the
SDE \begin{equation}
dW_{t}=(r-\al)W_{t}dt+\sigma W_{t}dB_{t}-Gdt,\quad0\leq t\leq T,\quad W_{0}=P.\label{eq:dWt}\end{equation}

An additional constraint must account for the non-negativity of the account value. As stated in \citep{RefWorks:7}, it can be shown that $W$ satisfies
\revised{
\begin{equation}
W_{t}^{x,u}=\left(xe^{(r-\al-\frac{1}{2}\sigma^{2})(t-u)+\sigma(B_{t}-B_{u})}-G\int_{u}^{t}e^{(r-\al-\frac{1}{2}\sigma^{2})(t-s)+\sigma(B_{t}-B_{s})}ds\right)^+,\label{eq:W}\end{equation}} where $(w)^+=\max(0,w)$.

Considering $W_t^{P,0}$, the initial premium $P$ can be factored out of \prettyref{eq:W}
because $G=gP=P/T$. Let $\{Z_{t}\}$ denote the account value process
under a no-withdrawal strategy beginning with $Z_{0}=1$. Then $Z_{t}$
follows the SDE \[
dZ_{t}=(r-\al)Z_{t}dt+\sigma Z_{t}dB_{t},\quad0\leq t\leq T,\quad Z_{0}=1,\]
with the solution \[
Z_{t}=e^{(r-\al-0.5\sigma^{2})t+\sigma B_{t}}.\]
 By \prettyref{eq:W} $W_{t}$
can be expressed in terms of $Z_{t}$: \begin{align}W_t=\max\left[0,PZ_{t}-G\int_{0}^{t}\frac{Z_{t}}{Z_{s}}ds\right].\label{eq:W1}\end{align}
\citep{RefWorks:7} use a slight variant of this expression involving
the inverse of $Z$.

\section{Valuation Perspectives and Decompositions}
\label{sec:val-persp}

There are two perspectives from which to view the GMWB rider. A policyholder
is likely to view the VA and rider as one combined instrument and
would be interested in the total payments received over the duration
of the contract. On the other hand, although the rider is embedded
into the VA, the insurer might want to consider it as a separate instrument.
Namely, the insurer is interested in mitigating and hedging the additional
risk attributed to the rider.

\subsection{Policyholder Valuation Perspective}
Using standard actuarial notation we write the present value
of a continuously paid term-certain annuity as \[
\ann{\ban}{T}=\int_{0}^{T}e^{-rs}ds=\frac{1-e^{-rT}}{r}.\]
Denote by $V_{0}$ the value at $t=0$ for the complete contract (VA
plus GMWB rider). The risk-neutral discounted value of the withdrawals and any terminal account value is 
\begin{equation}
V_{0}(P,\alpha,g)=E_{\Q}\left[\int_{0}^{T}Ge^{-rs}ds+e^{-rT}W_{T}\right]=G\ann{\ban}{T} +e^{-rT}E_{\Q}[W_{T}],\label{eq:V0}\end{equation}
as in \citep{RefWorks:7}. We write $V_{0}$ when the parameters are understood. Note that $V_{0}$ is an implicit function of the fee rate $\alpha$.

Let $\{V_{t}\}_{0\leq t\leq T}$ denote the value process
of the contract where $V_{t}$ considers only future cash flows occurring after time $t$, discounted
to time $t$, conditional on the information $\mc{F}_{t}$.
Then \begin{align}
V_{t} & =E_{\Q}\left[\int_{t}^{T}e^{-r(s-t)}Gds+e^{-r(T-t)}W_{T}|\mc{F}_{t}\right]\label{eq:Vi}\\
 & =G\ann{\ban}{T-t}+e^{-r(T-t)}E_{\Q}[W_{T}^{P,0}|\mc{F}_{t}].\nonumber \end{align}
By the Markov property for $W_{t}$ 
we have $V_{t}=v(t,W_{t}),\  \Q$-a.s., for all $t\in[0,T]$, where $v:[0,T]\times\mb{R}_{+}\mapsto\mb{R}_{+}$
is given by\[
v(t,x)=G\ann{\ban}{T-t}+e^{-r(T-t)}E_{\Q}[W_{T}^{x,t}].\]

\label{AlternFormsV}Alternatively, 
$V_{0}$ can be decomposed into the sum of a term certain annuity
component and either a Quanto Asian Put option on $Z^{-1}$ (see \citealp{RefWorks:7})
or an Asian Call (floating strike) option on $Z$ (see \citealp{RefWorks:43}).
In either formulation the value function $v$ must be a function of
both $Z_{t}$ and some functional $f(\{Z_{s};0\leq s\leq t\})$ because
only the joint process $\{Z_{u,}f(\{Z_{s};0\leq s\leq u\})$ is Markovian.
Therefore we choose to continue working directly with \prettyref{eq:V0}.
However, the alternative forms prove to be useful when we explore numerical implementations, a binomial model, and mortality diversification in companion papers \citep{GMWB-binomial,GMWB-mortality}

\begin{defn}
\label{def:fairfeerate}
A \textbf{fair fee rate} is a rate \revised{$\faira \geq 0$} such that
\begin{equation}
V_{0}(P,\faira,g)=P.\label{eq:!existalpha}\end{equation}
\end{defn}
 Equation \prettyref{eq:!existalpha}
does not have a closed form solution and numerical methods must be
used to find $\faira$. Since $P$ can be factored out of
\prettyref{eq:W} it follows that $V_{0}(P,\alpha,g)=PV_{0}(1,\alpha,g)$ and from
\prettyref{eq:!existalpha} it is immediate that $\faira$ is independent
of $P$. 
A key question is the existence and uniqueness results of the fair fee rate. Before answering this question we first consider the insurer's valuation problem.

\subsection{\label{sub:Insurer-Valuation}Insurer Valuation Perspective}

The alternative viewpoint, applicable to the insurer, is to explicitly
consider the embedded guarantee represented by the rider as
a standalone product. Recall the trigger
time defined by \citep{RefWorks:7}. 
\begin{defn}
\label{def:The-trigger-time}The \textbf{trigger time}, $\tau$, defined by
the stopping time\[
\tau:=\inf\{s\in[0,T];W_{s}^{P,0}=0\},\]
 is the first hitting time of zero by the account value process. The
convention $\inf(\emptyset)=\infty$ is adopted. If $\tau\leq T$ we say the option is triggered (or exercised) at
trigger time $\tau$.
\end{defn}
We have $W_{t}=0\text{ for all }t\geq\tau$. We define the respective non-decreasing sequences of stopping times
$\{\tau_{t}\}_{t\in[0,T]}$ and $\{\bar{\tau}_{t}\}_{t\in[0,T]}$
as $\tau_{t}:=\tau\vee t$
 and $\bar{\tau}_{t}:=\tau_{t}\wedge T,$
for all $t\in[0,T]$. \revised{Intuitively, if trigger occurred prior to $t$ then $\tau_{t}$ discards the information on the exact timing of the event.  
Further, since $\tau$ and $\tau_{t}$ could be infinite and the contract expires at time $T$, $\bar{\tau}_{t}$ is the minimum of $\tau_{t}$ and $T$.}

For $0\leq s\leq t\leq T$ and $A\subset[0,T]$,
by the Markov property of $W_{t}$ we have \begin{equation}
\Q(\bar{\tau}_{t}\in A|\mc{F}_{s})=F(s,t,A,W_{s}),\label{eq:Markov prop Qtau}\end{equation}
$\mb{\Q}$-a.s. where \[
F(s,t,A,w):=\mb{\Q}(\bar{\tau}_{t}^{w,s}\in A)\]
and \[
\bar{\tau}_{t}^{w,s}=\inf\{u\geq t;W_{u}^{w,s}=0\}\wedge T.\]

Let $U=\{U_{t};0\leq t\leq T\}$ denote the rider value process. \label{At-time-,}At time $\bartau0$
the rider guarantee entitles the policyholder to receive a term certain
annuity for $T-\bartau0$ years with an annual payment of $G$ and no uncertainty remains.
Fee revenue is received up to time $\bartau0$.

This motivates the following definition for $U$ which appears
in \citep{RefWorks:32}. For $t\in[0,T]$ we define\begin{gather}
U_{t}:=E_{\Q}\left[e^{-r(\bar{\tau}_{t}-t)}G\ann{\ban}{T-\bar{\tau}_{t}}-\int_{t}^{\bar{\tau}_{t}}e^{-r(s-t)}\alpha W_{s}^{P,0}ds|\mc{F}_{t}\right].\label{eq:Uicts}\end{gather}
The value $U_{t}$ represents the remaining risk exposure to the insurer and is the risk-neutral expected discounted difference
between future rider payouts and future fee revenues.
By the Markov property for $\{W_{t}\}$ and \prettyref{eq:Markov prop Qtau}
we have $U_{t}=u(t,W_{t}),\ \Q$-a.s. for all $t\in[0,T]$, where $u:[0,T]\times\mb{R}_{+}\mapsto\mb{R}$
is given by \begin{equation}
u(t,x)=E_{\Q}\Bigg[e^{-r(\bar{\tau}_{t}^{x,t}-t)}G\ann{\ban}{T-\bar{\tau}_{t}^{x,t}}-\int_{t}^{\bar{\tau}_{t}^{x,t}}e^{-r(s-t)}\alpha W_{s}^{x,t}ds\Biggr].\label{eq:UMarkov}\end{equation}
The boundary condition $u(t,0)=G\ann{\ban}{T-t}$ is implied
in the above formulation. 

\subsection{\label{sub:Analytic-Results}Analytic Results}

With the goal of arriving at an existence and uniqueness result for
$\alpha^{\star}$, we first state two basic properties satisfied by
$V_{0}$, the proofs of which can be found in the appendix.
\begin{lem}
\label{lem:Monot}$V_{0}$, defined by \prettyref{eq:V0}, is a strictly
decreasing and continuous function of $\alpha$ for $\alpha\geq0$. \end{lem}

\begin{thm}
\label{thm:UniqueAlpha}
Under \prettyref{ass:Cts_nolapses} there exists a unique $\alpha^{\star}$
satisfying\[
V_{0}(P,\alpha^{\star},g)=P.\]
\end{thm}

\begin{remUnnum}
\prettyref{ass:Cts_nolapses} imposed that $r>0$.  Otherwise, 
the optimal solution $\alpha^{\star}$ must satisfy $W_{T}(\alpha^{\star})=0$,
 $\Q$ a.s.\ and, by \prettyref{lem: probruin} of the Appendix, no solution exists.
\end{remUnnum}
The next result unifies the insured and insurer perspectives and was
first presented in \citep{RefWorks:32} for the case $t=0$ with
 stochastic interest rates. In \prettyref{sec:Extending-modelSurrenders} we extend
this result to the more general case of surrenders and a complete
proof will be presented at that time. 
\begin{prop}
\label{pro:U=00003DU+W}For any $\alpha\geq0$ \[
v(t,w)=u(t,w)+w\]
 for all $t\in[0,T]$ and $w>0$.
That is, $V_{t}=U_{t}+W_{t}$, $\Q$ a.s. \end{prop}
\begin{remUnnum}
By definition of the fair fee rate $\alpha^{\star}$ we have $U_{0}(P,\alpha^{\star},g)=0$
as a result of \prettyref{pro:U=00003DU+W}. From \prettyref{lem:Monot}
we have $V_{0}<P$ and $U_{0}<0$ for all $\alpha>\alpha^{\star}$.
Likewise, $V_{0}>P$ and $U_{0}>0$ for all $\alpha<\alpha^{\star}$.
For any $t$, we say the contract is in the money (ITM) if $V_{t}>W_{t}$
and $U_{t}>0$. Similarly, it is out of the money (OTM) if $V_{t}<W_{t}$
and $U_{t}<0$. It is at the money (ATM) if $V_{t}=W_{t}$ and $U_{t}=0$.
\end{remUnnum}

\begin{remUnnum}
In \prettyref{sub:Background} we briefly discussed the fund drag
created by an increase in the rider fee rate. The strictly decreasing
property of $V_{0}$ and \prettyref{pro:U=00003DU+W} imply that $U_{0}=V_{0}-P$
is a strictly decreasing function of $\alpha$. Thus any increase
in expected revenue from an increase in $\alpha$ will always exceed
any increase in expected rider payouts.
\end{remUnnum}

\section{Optimal Stopping and Surrenders}\label{sec:Extending-modelSurrenders}

We next extend the model to allow the policyholder to surrender, or lapse, the policy prior to time $T$. Although a policyholder may surrender for a number of reasons, for instance due to an emergency cash crisis, rational behaviour in an economic sense is assumed. Early surrenders occur only if the proceeds exceed the risk-neutral value of keeping the contract in-force. 

Upon surrender the policyholder closes out the contract by withdrawing
the current account value. The cash proceeds are reduced by a surrender
charge on any amount exceeding the annual maximal permitted withdrawal
amount specified in the rider contract. Typically, contract provisions
include contingent deferred sales charge (CDSC) schedules specifying
surrender charges as a function of the duration since issue year.
An example is an 8-year schedule with a charge of 8\% in year 1 and
decreasing by 1\% each year, followed by no surrender charges after
year 9. 

To describe the CDSC schedule let $k:[0,T]\mapsto[0,1]$ be a deterministic
non-increasing piecewise constant RCLL (right continuous with left
limits) function with possible discontinuities at integer time values.
Our results hold for any non-increasing function taking
values in $[0,1]$ but we select a function that is an accurate representation
of CDSC schedules for products sold in the insurance marketplace. For a policy issued at time zero, $k_{s}$ is the surrender charge
applicable at time $s$. The no-lapse model is easily recovered by
setting $k_{s}=1$ for all $s\in[0,T)$ and $k_{T}=0$ in which case
the opportunity to surrender early is worthless. Similarly, we could
model a contract which only allows surrenders once a specific duration
$t_{1}$ is reached, by setting $k_{s}=1$ for $s\in[0,t_{1})$ and
$k_{s}<1$ for $s\in[t_{1},T]$. However the more common case has
$k_{s}<1$ for all $s\in[0,T]$. Further, we assume $k_{T}=0$ to
allow comparison to the no-lapse model where the contract terminates
at time $T$ with no surrender charges.

We assume the proceeds from surrender charges are invested in the
hedging portfolio. Without surrender charges, it would be optimal
to surrender the contract when it is OTM and avoid paying future annual rider fees.
Surrender charges act as a disincentive and may make it too
costly to surrender, or if surrender is still optimal, they provide the insurer with compensation for
the loss of future fees.  For GMDB riders, \citep{RefWorks:35} argue {}``when
option premiums are paid by instalments - even in the presence of
complete mortality and financial markets - the ability to `lapse'
de facto creates an incomplete market''. In our view, the surrender charges complete the market and make the guarantees hedgeable.  In fact, in the case of a binomial model for the product without mortality we prove in \citep{GMWB-binomial} that the product can be perfectly hedged with fee income and surrender charges.  We also explore pricing, hedging errors, and diversification limits in a binomial model with mortality risk in \citep{GMWB-mortality}.

The pricing task, viewed from the policyholder's perspective, becomes an optimal stopping problem. The contract
value process for the VA plus GMWB is\begin{equation}
V_{t}:=\sup_{\eta\in\mb{L}_{t}}V_{t}^{\eta},\label{eq:VictsLapse}\end{equation}
where \begin{equation}
V_{t}^{\eta}=E_{\Q}\left[G\ann{\ban}{\eta-t}+e^{-r(\eta-t)}W_{\eta}(1-k_{\eta})|\mc{F}_{t}\right]\label{eq:VctslapseSpecificeta}\end{equation}
and $\mb{L}_{t}$ is the set of $\mb{F}-$adapted stopping
times taking values in $[t,T]$.
\revised{Observe that }it is sufficient to consider the set $\mb{L}_{t,\bar{\tau}_{t}}\subset\mb{L}_{t}$,
where $\mb{L}_{t,\bar{\tau}_{t}}$ contains all $\mb{F}-$adapted
stopping times taking values in $[t,\bar{\tau}_{t})\cup\{T\}$, and
$\bar{\tau}_{t}$ is the trigger time assuming no lapses.
That is, if the rider is triggered without prior surrender then \revised{due to product design} the
future guaranteed payments can not be immediately withdrawn and optimal
surrender will naturally occur at maturity time $T$.

By the Markov property of $W_{t}$ we have $V_{t}=v(t,W_{t})\ \Q-$a.s. for all $t\in[0,T]$, where $v:[0,T]\times\mb{R}_{+}\mapsto\mb{R}_{+}$
is given by\[
v(t,x)=\sup_{\eta\in\mb{L}_{t,\bar{\tau}_{t}^{x,t}}}E_{\Q}\left[G\ann{\ban}{\eta-t}+e^{-r(\eta-t)}W_{\eta}^{x,t}(1-k_{\eta})\right].\]
{}Suppose that $k_{0}=0$ and let $\hat{\alpha}:=\inf\{\alpha;V_{0}(P,\alpha,g)=P\}$.
Then for all $\alpha\geq\hat{\alpha}$ we have $V_{0}(P,\alpha,g)=P$,
but there will be no buyers as it is optimal to surrender immediately.
Insurers will not charge $\alpha<\hat{\alpha}$ because $V_{0}(P,\alpha,g)>P$.
When lapses are permitted but no surrender charges are imposed, there
is no unique $\faira$ and the product is not marketable. To preclude
this trivial case, we impose the condition that $k_{0}>0$. 

The insurer's value process for the rider, analogous to \prettyref{eq:Uicts}, is given by \begin{equation}
U_{t}:=\sup_{\eta\in\mb{L}_{t,\bar{\tau}_{t}}}U_{t}^{\eta},\label{eq:ULapse}\end{equation}
where \[
U_{t}^{\eta}=E_{\Q}\left[Ge^{-r(\bar{\tau}_{t}-t)}\mathbf{1}_{\{\eta=T\}}\ann{\ban}{T-\bar{\tau}_{t}}-\int_{t}^{\eta}\alpha e^{-r(s-t)}W_{s}ds-e^{-r(\eta-t)}W_{\eta}k_{\eta}|\mc{F}_{t}\right].\]

We introduce a value process for the option to surrender and denote
it by $L=\{L_{t};0\leq t\leq T\}$. Let $U_{t}^{NL}$ be the rider
value given by \prettyref{eq:Uicts} in the no-lapse model. Then we
define $L_{t}:=U_{t}-U_{t}^{NL}\geq0$ for all $t\in[0,T]$. It follows that \begin{equation}
L_{t}=\sup_{\eta\in\mb{L}_{t,\bar{\tau}_{t}}}L_{t}^{\eta},\label{eq:CtsLapse L}\end{equation}
where \[
L_{t}^{\eta}=E_{\Q}\left[\int_{\eta}^{T}\alpha e^{-r(s-t)}W_{s}ds-Ge^{-r(\bar{\tau}_{t}-t)}\mathbf{1}_{\{\eta<\bar{\tau}_{t}\}}\ann{\ban}{T-\bar{\tau}_{t}}-k_{\eta}W_{\eta}e^{-r(\eta-t)}|\mc{F}_{t}\right].\]
This formulation is quite intuitive. For a fixed surrender strategy,
the surrender benefit is the expected value of the fees
avoided by early surrender, less any future benefit payments missed
if surrender occurs prior to a trigger time, and less the surrender
charge paid at the time of surrender. It is natural that the insured
seeks to optimize this surrender benefit. The Markovian representations
for $U$ and $L$ are obvious and are omitted.

\prettyref{pro:U=00003DU+W} formalized the precise relationship between
$\{U_{t}\}$ and $\{V_{t}\}$ in the no-lapse model. The next theorem
generalizes that relationship to the current model and is an extension
of a result proved by \citep{RefWorks:32} in the case where no lapses are permitted.
\citep{RefWorks:32} considers a model with stochastic interest rates.  In contrast, we suppose interest rates are constant but generalize \citep[Eqn. (2.16)]{RefWorks:32} to all times $t\in[0,T]$ and allow for lapses.
\begin{thm}
\label{thm:VULWCts}Let $V_{t},U_{t}^{NL},L_{t},U_{t}$ be defined
by \prettyref{eq:VictsLapse}, \prettyref{eq:Uicts}, \prettyref{eq:CtsLapse L}
and \prettyref{eq:ULapse} respectively. Then for all $\alpha\geq0$
and for all $t\in[0,T]$, we have\begin{align}
V_{t} & =W_{t}+U_{t}^{NL}+L_{t},\quad\Q\ \text{a.s.},\label{eq:V++W+L}\end{align}
or, equivalently, \begin{equation}
V_{t}=W_{t}+U_{t},\quad\Q\ \text{a.s.}\label{eq:V+W+Ulapses}\end{equation}
\end{thm}
\begin{proof}
Fix $t\in[0,T]$. Applying the product rule to $(e^{-r(s-t)}W_{s})$
for any $s\in[t,T]$,\begin{align*}
d(e^{-r(s-t)}W_{s}) & =-re^{-r(s-t)}W_{s}ds+e^{-r(s-t)}dW_{s}\\
 & =-re^{-r(s-t)}W_{s}ds+e^{-r(s-t)}[(r-\alpha)W_{s}ds+\sigma W_{s}dB_{s}-Gds]\\
 & =-\alpha e^{-r(s-t)}W_{s}ds+e^{-r(s-t)}\sigma W_{s}dB_{s}-e^{-r(s-t)}Gds.\end{align*}

Fix $\eta\in\mb{L}_{t,\bar{\tau}_{t}}$. Integrating over the
interval $[t,\eta\wedge\bar{\tau}_{t}]$, and observing that $W_{s\wedge\bar{\tau}_{t}}=W_{s}$
for all $s\in[t,T]$, we obtain\[
e^{-r(\eta-t)}W_{\eta}-W_{t}=-\int_{t}^{\eta}\alpha W_{s}e^{-r(s-t)}ds-G\ann{\ban}{\eta\wedge\bar{\tau}_{t}-t}+\int_{t}^{\eta}e^{-r(s-t)}\sigma W_{s}dB_{s}.\]
Note that $G\ann{\ban}{\eta-t}=G\ann{\ban}{\eta\wedge\bar{\tau}_{t}-t}+Ge^{-r(\bar{\tau}_{t}-t)}\ann{\ban}{\eta\vee\bar{\tau}_{t}-\bar{\tau}_{t}}.$
Having fixed $\eta\in\mb{L}_{t,\bar{\tau}_{t}}$ we have $\ann{\ban}{\eta\vee\bar{\tau}_{t}-\bar{\tau}_{t}}=\mathbf{1}_{\{\eta=T\}}\ann{\ban}{T-\bar{\tau}_{t}}.$
Then \begin{multline*}
e^{-r(\eta-t)}W_{\eta}+G\ann{\ban}{\eta-t}=\\
W_{t}+Ge^{-r(\bar{\tau}_{t}-t)}\mathbf{1}_{\{\eta=T\}}\ann{\ban}{T-\bar{\tau}_{t}}-\int_{t}^{\eta}\alpha W_{s}e^{-r(s-t)}ds+\int_{t}^{\eta}e^{-r(s-t)}\sigma W_{s}dB_{s}.\end{multline*}
We have\revised{, from equation~(\ref{eq:W1}),} that \[
E_{\Q}\left[\int_{u}^{v}(W_{s})^{2}ds\right]<E_{\Q}\left[\int_{u}^{v}\revised{P^{2}}e^{2(r-\alpha-0.5\sigma^{2})s+2\sigma B_{s}}ds\right]<\infty,\]
thus by a standard result the above It\^{o} integral term is a martingale
(see \citep[Corollary 3.2.6]{RefWorks:2}) and $E_{\Q}[\int_{t}^{\eta}e^{-r(s-t)}\sigma W_{s}dB_{s}|\mc{F}_{t}]=0.$
Subtracting $e^{-r(\eta-t)}W_{\eta}k_{\eta}$ from both sides and
taking conditional expectations w.r.t. $\mc{F}_{t}$, we obtain\[
V_{t}^{\eta}=W_{t}+U_{t}^{\eta}.\]
Since $\eta$ was arbitrary, taking the supremum gives\[
V_{t}=W_{t}+U_{t}.\qedhere\]
\end{proof}
\begin{cor}
For any $\alpha\geq0$,
\begin{equation}L_{t} =\sup_{\eta\in\mb{L}_{t,\bar{\tau}_{t}}}E_{\Q}\left[e^{-r(\eta-t)}W_{\eta}(1-k_{\eta})-e^{-r(T-t)}W_{T}-Ge^{-r\eta}\ann{\ban}{T-\eta}|\mc{F}_{t}\right] \label{eq:LV}.\end{equation}
\end{cor}
\begin{proof}
\prettyref{pro:U=00003DU+W} and \prettyref{thm:VULWCts} imply $L_t=V_{t}-V_{t}^{NL}$ from which \prettyref{eq:LV} is obtained. 
\end{proof}

\begin{remUnnum}
For $\alpha^{\star}$, such that $V_{0}=P$, we have that $U_{0}(\alpha^{\star})=0$
and $L_{0}(\alpha^{\star})=-U_{0}^{NL}(\alpha^{\star})$. 
Equation \prettyref{eq:LV} is interpreted as the insured selecting the surrender
time to maximize the trade-off between receiving the account value
(less surrender charges) today, rather than at maturity, and foregoing
the rights to any future withdrawals.
\end{remUnnum}

\section{Conclusions}\label{sec:conclusions}

In this paper we have considered the valuation problem of a variable annuity with a guaranteed minimum withdrawal benefit  rider from the perspective of both the policy holder and the insurer.  The focus and main contributions of this paper are the financial aspects of the variable annuity and GMWB rider.   We define the fair rider fee as the rate which equates the risk-neutral expectation of all future benefits to the insured to the initial premium.  The first contribution of the paper is a proof of the existence and uniqueness of the fair fee.  The second contribution of the paper is an extension of the decomposition results of \citep{RefWorks:32} to include lapses.  We decompose the value of the contract into the account value, a component expressing the value of the guarantee in the no-lapse case, and a component expressing the value of the option to lapse.

The valuation perspectives and decompositions are expressed in terms of optimal stopping problems which require a numerical implementation technique to value the contract. \revised{We referred earlier to a subset of articles which focus on various approaches.  Rather than directly applying a numerical approach to the continuous time results presented in this paper, which are of independent interest, we focus on the fundamental theory here and place an emphasis on numerical implementation in our subsequent works.}  In \citep{GMWB-binomial} we construct a complete binomial framework for the contract recovering similar valuation perspectives and decompositions to those presented in this paper, and include hedging results for the GMWB rider as well as numerical implementation techniques. \revised{In particular, in \citep{GMWB-binomial}, we obtain numerical results based on the binomial model which can be viewed as an approximation to the model presented in this paper and which are in accordance with numerical results obtained using Monte-Carlo methods by \citep{Liu-Kolk} and \citep{RefWorks:43}.}     \revised{Further, w}hile we disregard mortality in this paper we extend the binomial modelling framework to include mortality and obtain numerical results demonstrating the limits of hedging and diversification of mortality risk in \citep{GMWB-mortality}.

\vspace{1em}
\noindent\textbf{Acknowledgements}
This research was supported by the Natural Sciences and Engineering Research Council (NSERC) of Canada and the Fonds de recherche du Qu\'ebec - Nature et technologies (FQRNT).  \revised{The authors thank the editor and anonymous referee for valuable comments which improved the paper.}

\appendix
\setcounter{thm}{0}
    \renewcommand{\thethm}{\Alph{section}\arabic{thm}}
\section{Additional Results and Proofs}
\label{Proofs}
This appendix is devoted to the proofs of \prettyref{lem:Monot} and \prettyref{thm:UniqueAlpha}.  The following two lemmas are required to prove \prettyref{lem:Monot}.

\begin{lem}
\label{lem:RuinProp}For any $T,a,k>0$ we have \textup{$\mathbb{Q}(\int_{0}^{T}e^{-aB_{s}}ds<k)>0,$}
where $B_{s}$ is a standard $\mathbb{Q}$-Brownian motion process\textup{}%
\textup{.}\end{lem}
\begin{proof}\footnote{The authors thank Dr.\ Anthony Quas, University of Victoria, for helpful discussion concerning the proof of \prettyref{lem:RuinProp}.}
\revised{Let $u=\min{\left(\frac{k}{2e^{a}},T\right)}$.}
Write $\int_{0}^{T}e^{-aB_{s}}ds=\int_{0}^{u}e^{-aB_{s}}ds+\int_{u}^{T}e^{-aB_{s}}ds.$
\revised{We consider the two cases $u=T$ and $u<T$ separately by conditioning} on the events $A=\{B_{s}>-1;\forall s\in[0,u]\}$
and $C=\{B_{s}>M;\forall s\in[u,T]\}$, where $M$ satisfies $e^{-aM}=\frac{k}{2(T-u)}$.  \\
\revised{\underline{(i) $u=T$}:} \\
\revised{Conditioning on the event $A$ we find that 
\[
\mathbb{Q}\left( \int_{0}^{T}e^{-aB_{s}}ds<k \right) \geq \mathbb{Q}\left( \int_{0}^{T}e^{-aB_{s}}ds<k | A \right) \mathbb{Q}(A) = \mathbb{Q}(A)
\]
since $A$ implies $\int_{0}^{T}e^{-aB_{s}}ds<Te^{a}\leq\frac{k}{2}$ and $\mathbb{Q}\left( \int_{0}^{T}e^{-aB_{s}}ds<k | A \right) = 1$.}  By \citep[formula 1.1.2.4]{RefWorks:42}
\begin{equation}
\revised{\mathbb{Q}}_{x}(\inf_{0\leq s\leq t}\tilde{B}_{s}>y)=2\Phi(\frac{x-y}{\sqrt{t}})-1,\quad y\leq x, \label{eq:a1}
\end{equation}
where $\Phi$ is the cdf of the standard normal distribution \revised{and $\tilde{B}_{t}$ is a Brownian motion with $\tilde{B}_0=x$ a.s.\ under $\mathbb{Q}_x$}.  \revised{Hence, equation~(\ref{eq:a1}), with $x=0$; $y=-1$; and $t=T$, implies that $\mathbb{Q}(A)>0$.  That is, $\mathbb{Q}(\int_{0}^{T}e^{-aB_{s}}ds<k)>0$ as desired.}

\noindent\revised{\underline{(ii) $u < T:$}}
\\
Conditioning on the events $A$ and $C$ we find that 
\begin{align*}
 \mathbb{Q}\left(\int_{0}^{T}e^{-aB_{s}}ds<k \right) &\geq \mathbb{Q}\left(\int_{0}^{u}e^{-aB_{s}}ds+\int_{u}^{T}e^{-aB_{s}}ds<k | \ A\cap C \right)\mathbb{Q}(A\cap C) \\ &= \mathbb{Q}(A\cap C)
\end{align*}
since $A$ implies $\int_{0}^{u}e^{-aB_{s}}ds<u e^{a} = \frac{k}{2}$ and $C$ implies $\int_{u}^{T}e^{-aB_{s}}ds<(T-u)e^{-aM}=\frac{k}{2}$.  
\revised{Similar to the first case we have $\mathbb{Q}(A) > 0 $ by equation~(\ref{eq:a1}).} 
\revised{To see $\mathbb{Q}(C\mid A)>0$, we introduce $D=\{B_u > M+\epsilon\}$ where $\epsilon >0$ and $B_0=0$. Then $\mathbb{Q}(D)>0$. We could use reflection-type arguments to show that $\mathbb{Q}(A\mid D)>0$ but we refer to a more complete result in \citep[Proposition 4.3.5.3]{RefWorks:37} on the maximum of a general Brownian bridge. Conditioning on $D$ and taking expectations w.r.t. $B_u$, that result can be used to derive an explicit expression for $\mathbb{Q}(A\mid D)$. Therefore $\mathbb{Q}(D\mid A)>0$.}

\revised{Finally, equation~(\ref{eq:a1}), with $x=M+\epsilon$; $y=M$; and $t=T-u$, implies $\mathbb{Q}(C\mid A\cap D)>0$. Thus 
$$\mathbb{Q}(A\cap C)=\mathbb{Q}(A)\mathbb{Q}(C\mid A)\geq \mathbb{Q}(A)\mathbb{Q}(C\mid A\cap D)\mathbb{Q}(D\mid A)>0.$$}

\end{proof}

\begin{lem}
\label{lem: probruin}For any fee rate $\al$ and guaranteed withdrawal
rate $g$ there is a positive probability that the contract matures
with a positive account value. That is, \[
\Q(W_{T}^{P,0}>0)>0\]
 for all $P>0$, $g>0$, and $\al\geq0$, where $W_{T}^{P,0}$
is given by \prettyref{eq:W}. \end{lem}
\begin{proof}
Note that $W_{T}^{P,0}>0$ if and only if \[\frac{P}{G}>\int_{0}^{T}e^{-(r-\al-0.5\sigma^{2})s-\sigma B_{S}}ds.\]
By bounding and removing the deterministic portion from the integrand,
we have\[
\frac{P}{G}>\int_{0}^{T}e^{-(r-\al-0.5\sigma^{2})s-\sigma B_{S}}ds\]
if \[
\frac{P}{G}c^{-1}>\int_{0}^{T}e^{-\sigma B_{s}}ds,\]
 where \[
c=\begin{cases}
e^{-(r-\al-0.5\sigma^{2})T} & \text{if }(r-\al-0.5\sigma^{2})<0,\\
1 & \text{otherwise}.\end{cases}\]
The desired conclusion follows from that fact that $\Q(\int_{0}^{T}e^{-aB_{s}}ds<k)>0$ for all $T,\ a,\ k>0$ of \prettyref{lem:RuinProp}. 
\end{proof}

\begin{proof}[Proof of \prettyref{lem:Monot}]
We fix $P$ and $g$ and omit them from the notation. A monotonicity
result is obtained by applying a comparison result for SDEs from \citep[Proposition 2.18]{RefWorks:13}.
Since $\alpha$ appears as a negative drift term in the SDE for $W_{t}$
in \prettyref{eq:dWt}, we have $W_{t}(\alpha_{1})\geq W_{t}(\alpha_{2})$
a.s.\ for all $t\in[0,T]$ and for all $0\leq\alpha_{1}<\alpha_{2}$.
Thus $E_{\Q}[W_{T}(\alpha_{1})]\geq E_{\Q}[W_{T}(\alpha_{2})]$ which
implies $V_{0}(\alpha_{1})\geq V_{0}(\alpha_{2})$.

To prove the strictly decreasing property of $V_{0}(P,\alpha,g)$ note from \prettyref{lem: probruin} that $\Q(A^{\alpha})>0\text{ for all }\alpha\geq0$ 
where $A^{\alpha}:=\{W_{T}(\alpha)>0\}$.
On the event $A^{\alpha}$ we have \[
W_{T}(\alpha)=e^{(r-\alpha-0.5\sigma^{2})T+\sigma B_{T}}\times\left(P-G\int_{0}^{T}e^{-(r-\alpha-0.5\sigma^{2})s-\sigma B_{s}}ds\right).\]
Let $0\leq\alpha_{1}<\alpha_{2}=\alpha_{1}+h,$ where $h$ takes an
arbitrary positive value. Restricted to the set $A^{\alpha_{1}+h}$,
we obtain\[
W_{T}(\alpha_{1}+h)\leq e^{-hT}W_{T}(\alpha_{1})<W_{T}(\alpha_{1})\]
implying that $A^{\alpha_{1}}\supseteq A^{\alpha_{1}+h}$. It follows
that \begin{align*}
V_{0}(\alpha_{1}+h) & =G\bar{a}_{\lcroof{T}}+E_{Q}\left(e^{-rT}W_{T}(\alpha_{1}+h)\mathbf{1}_{\left\{ A^{\alpha_{1}+h}\right\} }\right)\\
 & <G\bar{a}_{\lcroof{T}}+E_{Q}\left(e^{-rT}W_{T}(\alpha_{1})\mathbf{1}_{\left\{ A^{\alpha_{1}+h}\right\} }\right)\\
 & \leq V_{0}(\alpha_{1}).\end{align*}

To prove continuity fix $\alpha\geq0$. Let $h>0$ and denote\[
X_{T}^{h}:=e^{\sigma B_{T}}\max\left(0,P-G\int_{0}^{T}e^{-(r-\alpha-h-\frac{1}{2}\sigma^{2})s-\sigma B_{s}}ds\right).\]
From \prettyref{eq:W}, \[
E_{Q}(W_{T}(\alpha+h))=e^{(r-\alpha-h-\frac{1}{2}\sigma^{2})T}E_{Q}\left(X_{T}^{h}\right).\]
Then $X_{T}^{h}\geq0$ for all $h\geq0$, and $X_{T}^{h}\uparrow$
a.s.\ as $h\downarrow0$. Applying the Monotone Convergence theorem
and by the continuity of the max function, \[
\lim_{h\downarrow0}E_{\Q}(X_{T}^{h})=E_{\Q}(X_{T}^{h=0}).\]
The Dominated Convergence theorem was used to interchange the limit
and the path-wise Lebesgue-Stieltjes integral. Therefore $\lim_{h\downarrow0}E_{\Q}(W_{T}(\alpha+h))=E_{\Q}(W_{T}(\alpha))$. 

If $\alpha>0$, then let $h<0$ and $\lim_{h\uparrow0}E_{\Q}(W_{T}(\alpha+h))=E_{\Q}(W_{T}(\alpha))$
is obtained using similar arguments. The Monotone Convergence theorem
no longer applies; instead the Dominated Convergence theorem justifies
interchanging the expectation and limit since $X_{T}^{h}\leq Pe^{\sigma B_{T}}$
and $E_{\Q}(e^{\sigma B_{T}})=e^{0.5\sigma^{2}T}<\infty$. Therefore
the continuity of $V_{0}$ follows from \prettyref{eq:V0}.\end{proof}

\begin{proof}[Proof of \prettyref{thm:UniqueAlpha}]
The existence of $\alpha^{\star}$ is obtained by showing that both
$V_{0}(P,0,g)\geq P$ and $\lim_{\alpha\to\infty}V_{0}(P,\alpha,g)<P$
and applying the continuity result from \prettyref{lem:Monot}.

When $\alpha=0$, the guarantee is offered at no charge and it is
obvious that $V_{0}\geq P$. More formally, setting $\alpha=0$ we
have from \prettyref{eq:W} \[
W_{T}\geq\left[Pe^{(r-0.5\sigma^{2})T+\sigma B_{T}}-G\int_{0}^{T}e^{(r-0.5\sigma^{2})(T-s)+\sigma(B_{T}-B_{s})}ds\right],\]
 and since $E_{\Q}[e^{-0.5\sigma^{2}t+\sigma B_{t}}]=1$, we obtain
from \prettyref{eq:V0} that\begin{align*}
V_{0}(P,0,g) & \geq P+E_{\Q}\left[\int_{0}^{T}e^{-rs}G\left(1-e^{-(0.5\sigma^{2})(T-s)+\sigma(B_{T}-B_{s})}\right)ds\right] %
=P,\end{align*}
 where the expectation on the right evaluates to zero by Fubini's
theorem. 

As $\alpha\to\infty$, it becomes certain that the embedded GMWB option
will be exercised and thus $V_{0}=G\bar{a}_{\lcroof{T}}$. More formally,
for $\alpha>0$ we have \begin{equation}
0\leq W_{T}(\alpha)\leq Pe^{-\alpha T}e^{(r-0.5\sigma^{2})T+\sigma B_{T}}\leq Pe^{(r-0.5\sigma^{2})T+\sigma B_{T}}\label{eq:Uniqalphainequ}\end{equation}
 a.s., and $E_{\Q}[Pe^{(r-0.5\sigma^{2})T+\sigma B_{T}}]=Pe^{rT}<\infty$.
The property $B_{T}<\infty$ a.s.\ combined with \prettyref{eq:Uniqalphainequ}
gives $\underset{\alpha\to\infty}{\lim}W_{T}(\alpha)=0$ a.s. Applying
the Dominating Convergence theorem, \[
\lim_{\alpha\to\infty}V_{0}(P,\alpha,g)=G\int_{0}^{T}e^{-rs}ds<GT=P,\]
for $r>0$.

The uniqueness of the solution follows directly from the strictly
decreasing property for $V_{0}(P,\alpha,g)$ from \prettyref{lem:Monot}. \end{proof}

\bibliographystyle{model2-names}
\bibliography{wmendy-RefList1}
\end{document}